\documentclass[11pt]{article}
\usepackage{path,graphicx}
\usepackage{a4,amssymb, dsfont}
\usepackage{authblk}
\usepackage{amsmath}
\usepackage{cite}
\usepackage[]{algorithm2e}
\usepackage[margin=1in]{geometry}
\usepackage[justification=centering]{caption}
\newtheorem{definition}{Definition}[section]
\newtheorem{lemma}[definition]{Lemma}
\newtheorem{theorem}[definition]{Theorem}

\newcommand{\qed}{\hfill $\square$\smallskip}
\newcommand{\N}{\mathds{N}}
\title{Dynamic Monopolies in Reversible Bootstrap Percolation}
\author{Clemens Jeger}
\affil{Department of Mathematics, ETH Zurich}
\author{Ahad N. Zehmakan}
\affil{Department of Computer Science, ETH Zurich}
\providecommand{\keywords}[1]{\textbf{\textit{Key Words:}} #1}
\date{} 
\begin{document}
\maketitle

\begin{abstract}
We study an extremal question for the (reversible) $r-$bootstrap percolation processes. Given a graph and an initial configuration where each vertex is active or inactive, in the $r-$bootstrap percolation process the following rule is applied in discrete-time rounds: each vertex gets active if it has at least $r$ active neighbors, and an active vertex stays active forever.
In the reversible $r$-bootstrap percolation, each vertex gets active if it has at least $r$ active neighbors, and inactive otherwise. We consider the following question on the $d$-dimensional torus: how many vertices should be initially active so that the whole graph becomes active? Our results settle an open problem by Balister, Bollobas, Johnson, and Walters~\cite{balbol09} and generalize the results by Flocchini, Lodi, Luccio, Pagli, and Santoro~\cite{flocchini04}.
\end{abstract}
\keywords{Bootstrap percolation, dynamic monopoly, r-threshold model, majority rule, percolating set}
\section{Introduction}
There are many real-life phenomena that can be modeled as discrete dynamical systems of the following type: there is a set of agents with a neighborhood relation, and in each time step, the state of every agent is updated according to the states of its neighboring agents. Often, one is interested in understanding the long-term behavior of the system.

An example is \emph{rumor spreading} in a social network. Here, a ``rumor'' just stands for a piece of information that agents may or may not have; for example, knowledge of a specific video. A similar example is \emph{opinion spreading}. In the simplest case, every agent has at each time one of two opposite opinions but may change back and forth several times, depending on the opinions in the neighborhood. In both examples, states of agents are binary, and we distinguish \emph{active} agents (who know the rumor, or have positive opinion), from the \emph{inactive} ones. 

In this binary scenario, a specific question concerning the long-term behavior of such a system is whether it \emph{percolates}; i.e.\ whether eventually, all agents will be active. This is motivated by the desire to understand for example under which conditions a video becomes viral, or a specific opinion becomes dominant.

There are several reasonable ways of modeling the influence of neighboring agents. In this paper, we focus on three well-studied models. In each case, the ``agents with neighborhood relation'' form an undirected graph $G$, where initially, each vertex is either active or inactive. In discrete time rounds, vertices simultaneously update their state according to some fixed and deterministic rule. 


In the \emph{r-bootstrap percolation} (or shortly $r-$BP) process, an inactive vertex with at least $r$ active neighbors becomes active and stays active forever. This models rumour spreading processes. 



In the \emph{reversible r-bootstrap percolation}, a vertex becomes active if it has at least $r$ active neighbors, and inactive otherwise. This models opinion spreading processes. Furthermore, in the \emph{majority model} each vertex updates its state (active or inactive) to the most frequent state in its neighborhood, and keeps its current state in case of a tie. This can be considered as a variant of the reversible r-bootstrap percolation where each vertex $v$ has an individual threshold $r_v$ depending on its degree in the graph. In the previous two models, $r$ is a fixed parameter and the same for all vertices.


To understand percolation, we use the very well-studied concept of a dynamic monopoly. A \emph{configuration} is an assignment of active/inactive states to the vertices. In (reversible) $r-$BP or majority model, a configuration is called a \emph{dynamic monopoly}, or shortly \emph{dynamo}, if by starting from this configuration, the graph eventually gets fully active. Furthermore, the dynamo is \emph{monotone} if no active vertex ever becomes inactive during the process. Clearly, every dynamo in $r-$BP is also a monotone dynamo. Furthermore, in the reversible variant, each monotone dynamo is a dynamo, but not necessarily the other way around. We are interested in the size of a smallest dynamo, where the size of a configuration is the number of its active vertices. In the concrete applications, we are for example trying to understand how many (and which) initial ``influencers'' it takes for a video to become viral, or for an opinion to become dominant.

Although the concept of a dynamo was studied earlier, e.g. by Schonmann~\cite{schonmann1992behavior} and Balogh and Pete~\cite{balogh1998random}, it was formally defined and studied in the seminal work by Kempe, Kleinberg, and Tardos~\cite{optimization_maximize_sprad_of_influence} and independently by Peleg~\cite{PELEG_dynamo}, motivated from fault-local mending in distributed systems. There is a massive body of work concerning the minimum size of a (monotone) dynamo for (reversible) $r-$BP and majority model in different classes of graphs, for instance hypercube~\cite{balogh2010bootstrap,hambardzumyan2017polynomial,morrison2018extremal}, the binomial random graph~\cite{threshold_dynamo_g_n_p_2013,threshold_dynamo_random_graph_2017,chang2011spreading}, random regular graphs~\cite{gartner2018majority}, power-law graphs~\cite{dynamo_approximation_hard}, planar graphs~\cite{peleg2002local}, and many others. Motivated from the literature of cellular automata and a range of applications in statistical physics special attention has been devoted to the $d$-dimensional torus $\mathbb{T}_n^d$; the $d-$dimensional torus $\mathbb{T}_n^d$ is the graph with vertex set $[n]^d$, where two vertices are adjacent if and only if they differ by $1$ mod $n$ in exactly one coordinate (see Definition~\ref{def:torus}). 

\textbf{Our contribution:} Prior works, for instance see~\cite{balogh1998random,balbol09,riedl2010largest,hambardzumyan2017polynomial,flocchini04}, have established some lower and upper bounds on the minimum size of a (monotone) dynamo in $\mathbb{T}_n^d$ for the above models; however, several questions have remained unanswered. The main goal of the present paper is to settle these questions and provide a complete picture (see Table~\ref{Table 1}). In particular, our results settle an open problem by Balister, Bollobas, Johnson, and Walters~\cite{balbol09} and generalize the results by Flocchini, Lodi, Luccio, Pagli, and Santoro~\cite{flocchini04}. We sometimes shortly write small $r$ and large $r$ respectively for $1\leq r\leq d$ and $d+1\leq r\leq 2d$.   


\begin{table}[!ht]
\centering
    \begin{tabular}{| l | l | l |}
    \hline
    \textbf{process} & \textbf{dynamo} & \textbf{monotone dynamo} \\
    \hline
    \hline & &\\
    $r-$BP  (small $r$) & \boldmath{$\frac{1}{r}\binom{d}{r-1} n^{r-1}\pm\Theta(n^{r-2})$}~\cite{morrison2018extremal} & \boldmath{$\frac{1}{r}\binom{d}{r-1} n^{r-1}\pm\Theta(n^{r-2})$}~\cite{morrison2018extremal} \\
    
    & &\\ \hline & &\\
    $r-$BP  (large $r$) & $(1-\frac{d}{r})n^d\pm\Theta(n^{d-1})$~\cite{balogh1998random,balbol09} & $(1-\frac{d}{r})n^d\pm\Theta(n^{d-1})$~\cite{balogh1998random,balbol09} \\
    
    & &\\ \hline & &\\
    reversible $r-$BP  (small $r$) & \boldmath{$\frac{1+i}{r}\binom{d}{r-1}n^{r-1}\pm\Theta(n^{r-2})$} & \boldmath{$\frac{2}{r}\binom{d}{r-1}n^{r-1}\pm\Theta(n^{r-2})$} \\
    
        & \scriptsize{\boldmath{$i=1$ \textbf{if} $n$ \textbf{is even,} $i=0$ \textbf{o.w.}}} & \\

    & &\\ \hline & &\\
    reversible $r-$BP  (large $r$) & \boldmath{$2(1-\frac{d}{r})n^d\pm\Theta(n^{d-1})$} & \boldmath{$2(1-\frac{d}{r})n^d\pm\Theta(n^{d-1})$} \\
    
    & &\\ \hline & &\\
    majority model & $(1-\frac{d}{d+1})n^d\pm\Theta(n^{d-1})$~\cite{balbol09} & \boldmath{$(1-\frac{d}{d+2})n^d\pm\Theta(n^{d-1})$} \\

    & & \\
    \hline
    \end{tabular}
\caption{ The minimum size of a (monotone) dynamo in $\mathbb{T}_{n}^{d}$. Our results are marked in bold. However, the lower bound for $r-$BP (small $r$) comes from~\cite{morrison2018extremal}, and some of our results (as will be discussed in the text) were known for $d=2$ and $r=d,d+1$ respectively by~\cite{flocchini04,balogh1998random} and~\cite{balbol09}.} \label{Table 1}
\end{table}
In the rest of this section, we elaborate on the entries of Table~\ref{Table 1} and provide the main ideas behind our proof techniques. Let us first consider the size of a smallest dynamo in $r-$BP, which is supposed to be easier to handle due to its monotonicity. For large $r$, by applying the invariant method Balogh and Pete~\cite{balogh1998random} proved the lower bound of $(1-\frac{d}{r})n^d-\Theta(n^{d-1})$ on the minimum size of a dynamo in $\mathbb{T}_n^d$ and conjectured that this bound is tight (all results discussed in this paper are tight up to the terms of lower orders, here $\Theta(n^{d-1})$), which was verified several years later by Balister et al.~\cite{balbol09}. For small $r$, Balogh and Pete~\cite{balogh1998random} claimed the size of a smallest dynamo is of order $\Theta(n^{r-1})$ and provided tight bounds for $r=2$ and $r=d$. After several unsuccessful attempts (see e.g.~\cite{riedl2010largest,balogh2006bootstrap}), very recently Morrison and Noel~\cite{morrison2018extremal} proved the lower bound of $\frac{1}{r} {d\choose r-1} n^{r-1}$. They rely on some prior results regarding the weak saturation number of $\mathbb{T}_{n}^{d}$. Building on their results, Hambardzumyan, Hatami, and Qian~\cite{hambardzumyan2017polynomial} proved the upper bound of ${d\choose r-1}n^{r-1}$, which leaves an $r$ factor gap. To close this gap, we construct dynamos whose size matches the lower bound, up to the terms of lower orders.

On the other hand, there is not much known concerning the size of a smallest dynamo in reversible $r-$BP and $\mathbb{T}_n^d$. Balister et al.~\cite{balbol09} studied the problem for $r=d,d+1$ and posed the case of $r>d+1$ as an open problem. To address their problem, we prove that the size of a smallest dynamo in reversible $r-$BP is $2(1-\frac{d}{r})n^{d}\pm \Theta(n^{d-1})$ for $d+1\leq r\leq 2d$. Comparing this result with the one regarding $r-$BP implies the number of active vertices needed to activate the whole torus $\mathbb{T}_n^d$ for large $r$ and reversible $r-$BP is two times for $r-$BP. 

We also study reversible $r-$BP for small $r$, and interestingly, it turns out in this setting the parity of $n$ is a determining factor. More precisely, we prove the size of a smallest dynamo in $\mathbb{T}_n^d$ and under reversible $r-$BP for small $r$ is equal to $\frac{1}{r} {d \choose r-1} n^{r-1}\pm\Theta(n^{r-2})$ and $\frac{2}{r} {d \choose r-1} n^{r-1}\pm\Theta(n^{r-2})$ respectively for odd $n$ and even $n$. As a simple example, for $r=1$ and $\mathbb{T}_n^1$, which corresponds to the cycle $C_n$, a single active vertex makes the cycle fully active for odd $n$, but for even $n$ one needs two active vertices. 

    
    
Furthermore, we provide tight bounds on the size of a smallest monotone dynamo in reversible r-BP, while the answer previously was known only for the special case of $r=d=2$ by Flocchini et al.~\cite{flocchini04}. It is shown that the minimum size of a monotone dynamo in reversible $r-$BP and $\mathbb{T}_n^d$ is  $\frac{2}{r} {d \choose r-1} n^{r-1}\pm\Theta(n^{r-2})$ and $2(1-\frac{d}{r})n^{d}\pm\Theta(n^{d-1})$ respectively for $1\le r\le d$ and $d+1\le r\le 2d$, regardless of the parity of $n$. Interestingly, it implies that for large $r$ dropping the monotonicity constraint has almost no impact, in the sense that the size of a smallest monotone dynamo and dynamo are the same, up to the terms of lower orders.

Our upper bounds are constructive, meaning we provide dynamos of corresponding sizes and argue that they eventually make the whole torus active. For large $r$ we present an explicit subset of vertices which has our desired size and discuss if all vertices in this set are initially active then some other subsets of vertices get active one by one until the whole torus is active. For small $r$, we show if a specific group of $(r-1)-$dimensional sub-tori are fully active then the whole torus gets active and for the activation of these sub-tori we rely on our results for large $r$. Regarding the lower bounds, we prove for any bipartite graph (with minimum degree at least $r$) the size of a smallest dynamo under reversible $r-$BP is at least twice as large as the size of a smallest dynamo under $r-$BP. The proof of this statement is built on the simple fact that for a bipartite graph $G=(V\cup U,E)$ and the reversible $r-$BP, the states of the vertices in $V$ in the next round only depend on the states of the vertices in $U$ in this round and vice versa. Putting this statement in parallel with the lower bounds from $r-$BP and applying some small modifications provide the desired lower bounds.

In the majority model, Balister et al.~\cite{balbol09} proved the minimum size of a dynamo is equal to $(1-\frac{d}{d+1})n^{d}\pm\Theta(n^{d-1})$. Flocchini et al.~\cite{flocchini04} studied the monotone case for $d=2$. We extend their results by showing that the size of a smallest monotone dynamo in the majority model is $(1-\frac{d}{d+2})n^{d}\pm\Theta(n^{d-1})$. These results lead to the following interesting observation. In $\mathbb{T}_n^d$, the only difference among the majority model and reversible $r-$BP with $r=d$, which is also known as the \emph{biased majority model}~\cite{SCHONMANN1990619}, is the tie-breaking rule; in case of a tie in the biased variant, a vertex always gets active while in the majority model it stays unchanged. As discussed, the size of a smallest (monotone) dynamo in the majority model is of order $\Theta(n^d)$ while it is of order $\Theta(n^{d-1})$ for the biased majority. This demonstrates how minor local alternations can result in substantial changes in the global behavior of the process.


Motivated from modeling of the behavior of certain interacting particle systems like fluid flow in rocks~\cite{adler1988diffusion} and also some biological interactions \cite{cardelli2012cell}, these processes also have been extensively studied on $\mathbb{T}_n^d$ when the initial configuration is at random, that is each vertex is independently active with some probability $p$, and inactive otherwise. Roughly speaking, the main result is that there is some threshold value $p_c$ such that if $p$ is ``sufficiently'' larger (similarly smaller) than $p$ then the process gets (resp. does not get) fully active with a probability tending to one. In other words, $p_c$ is the threshold for the initial configuration to be a dynamo. For $r-$BP, after a large series of papers, e.g. see the results by Aizenman and Lebowitz~\cite{aizenman1988metastability}, Cerf and Cirillo~\cite{cerf1999finite}, and Holroyd~\cite{holroyd2003sharp}, finally Balogh, Bollobas, Duminil-Copin, and Morris~\cite{bootstrap_percolation_sharp_threshold_lattice} found the exact value of $p_c$. However, for reversible $r-$BP and the majority model the value of $p_c$ is known only for the special case of $d=2$~\cite{SCHONMANN1990619,Gaertner2017_majority_torus,gartner2017biased}. Our results might shed some light on how to find the value of $p_c$ for all dimensions. 

To recapitulate, we study the minimum size of a (monotone) dynamo in $r-$BP, reversible $r-$BP, and the majority model on the $d-$dimensional torus $\mathbb{T}_n^d$. Some lower and upper bounds are known by prior results, but several gaps and open problems have been left. Applying new techniques, we settle some open problems, including a problem by Balister et al.~\cite{balbol09}, and generalize some prior results, including the results by Flocchini et al.~\cite{flocchini04}. All in all, our results combined with prior results provide a complete picture of the minimum size of a (monotone) dynamo in $\mathbb{T}_n^d$ for all the aforementioned processes.

\subsection{Preliminaries and Basic Definitions}
Let $G=(V,E)$ be a simple graph with vertex set $V=\{v_1,\dots,v_N\}$ and let $\Omega=\{0,1\}^N$ be our binary state space. To the vertices of $G$ we assign an initial \emph{configuration} $\omega=(\omega_1,\dots,\omega_N)$ such that every vertex $v_i\in V$ has an initial \emph{state} $\omega_i\in\{0,1\}$. We say a vertex is \emph{active} if its state is 1, and \emph{inactive} if its state is 0.
Every function $f=f_G:\Omega\to\Omega$ that deterministically assigns a new state to every vertex in $G$ depending only on the states of its neighbors (plus itself) is called an \emph{updating rule}, where the \emph{neighborhood} of a vertex $v\in V$ is defined to be $N(v):=\{u\in V:\{v,u\}\in E\}$. Given an updating rule $f$ on $G$ and an initial configuration $\omega\in\Omega$, the updating rule $f$ induces a discrete time \emph{process} $\{\omega^{(t)}\}_{t\in\N}$ with $\omega^{(0)}\equiv\omega$ and $\omega^{(t)}=f(\omega^{(t-1)})$ for all $t\in\N$. 


\begin{definition}[(reversible) $r-$BP rule]\label{def:irreversible_r-threshold_rule}
Let $r\in\N$ and $\omega\in\Omega$ a configuration on graph $G=(V,E)$.
The \emph{$r-$BP rule} $f(\cdot,r):\Omega\to\Omega$ is defined via
 \begin{equation*}\label{eq:irrev_r-threshold_rule}
    f(\omega_i,r):=\begin{cases}1 & \text{if }\sum\limits_{v_j\in N(v_i)} \omega_j\geq r\\ \omega_i & \text{otherwise}\end{cases},~~i=1,\dots,N.
\end{equation*}
The \emph{reversible $r-$BP rule} $\overleftarrow{f}(\cdot,r):\Omega\to\Omega$ is defined via
\begin{equation*}\label{eq:r-threshold_rule}
    \overleftarrow{f}(\omega_i,r):=\begin{cases}1 & \text{if }\sum\limits_{v_j\in N(v_i)} \omega_j\geq r\\ 0 & \text{otherwise}\end{cases},~~i=1,\dots,N.
\end{equation*}

\end{definition}
Given a graph with initial configuration $\omega^{(0)}$, $r-$BP rule then induces an \emph{$r-$BP process} $\{\omega^{(t)}\}_{t\in\N}$ via $\omega^{(t)}:=f(\omega^{(t-1)},r)$ for $t\in\N$, and analogously $\overleftarrow{f}$ induces a \emph{reversible $r-$BP process}. A somewhat similar rule is the \emph{majority rule}. This rule assigns to every vertex the most frequent state in its neighborhood, with self-preference in case of a tie.
\begin{definition}[majority rule]
For $\omega\in\Omega$, the \emph{majority rule} $\overleftarrow{f}(\cdot,maj):\Omega\to\Omega$ is given as
\begin{equation*}\label{eq:maj_rule}\overleftarrow{f}(\omega_i,maj):= \begin{cases} 
      1 & \text{if}\sum\limits_{v_j\in N(v_i)} \omega_j+\frac{w_i}{2} > |N(v_i)|/2 \\
      0 & \text{otherwise} 
   \end{cases}
,~~i=1,\dots,N.\end{equation*}
\end{definition}
The term $\frac{\omega_i}{2}$ is to make sure that in case of a tie a vertex preserves its state. As before, we say this rule induces a \emph{majority process}. 

For $\omega,\omega'\in\Omega=\{0,1\}^N$, we have $\omega\geq\omega'\Leftrightarrow \omega_i\geq \omega'_i~\forall i=1,\dots,N$. Furthermore, we denote $|\omega|:=\sum_{i=1}^{N}|\omega_i|$.



Let us now formally define our main object of interest, dynamo.

\begin{definition}[(monotone) dynamo]\label{def:monotone_dynamo}
Given a graph $G=(V,E)$ and an updating rule $f$, a \emph{dynamo} under $f$ is an initial configuration $\omega^{(0)}\in\Omega$ such that after some time $t\in\N$ in the process $\{\omega^{(t)}\}_{t\in\N}$ (induced by $f$ and $\omega^{(0)}$) we have $|\omega^{(t')}|=|V|$ for all $t'\geq t$. If $\omega^{(0)}$ is a dynamo and $\omega^{(t+1)}\geq\omega^{(t)}$ for all $t\geq 0$ then $\omega^{(0)}$ is called \emph{monotone}.
\end{definition}

Let $G$ be a graph and $r\in\N$ fixed. We denote
\begin{itemize}
    \item $m(G,r):=min\{|\omega|:\omega\text{ is a dynamo in $r-$BP on $G$}\}$
    \item $\overleftarrow{m}(G,r):=min\{|\omega|:\omega\text{ is a dynamo under reversible $r-$BP on $G$}\}$
    \item $\overleftarrow{m}_{mon}(G,r):=min\{|\omega|:\omega\text{ is a monotone dynamo under reversible $r-$BP on $G$}\}$.
\end{itemize}

By definition, a monotone dynamo is also a dynamo in reversible $r-$BP. Furthermore, a dynamo in reversible $r-$BP is a dynamo in $r-$BP. Thus, $m(G,r)\le \overleftarrow{m}(G,r)\le \overleftarrow{m}_{mon}(G,r)$.

\begin{definition}[stable]
Let $\{\omega^{(t)}\}$ be the process induced by updating rule $f$ on graph $G$ and the initial configuration $\omega^{(0)}$. We say an active vertex (analogously, a subset of active vertices or its induced subgraph) is \emph{stable} in configuration $\omega^{(t')}$ for some $t'\geq 0$ if it stays active in all upcoming configurations, i.e. $\omega^{(t)}$ for $t\geq t'$. 
\end{definition}
Notice in $r-$BP, any set of active vertices in a configuration is stable. In reversible $r-BP$ and a $2d$-regular graph, including $\mathbb{T}_n^d$, for some set of active vertices $S$ in a configuration $\omega$ if each vertex in $S$ has at least $r$ neighbors in $S$ then $S$ is stable. Furthermore, let us state the following simple lemma which comes in handy several times later on.

\begin{lemma}
\label{stable-lemma}
In reversible $r$-BP or the majority model and a graph $G$, if for some initial configuration $\omega$ all active vertices are stable, every vertex which gets active during the process is stable.  
\end{lemma}
This is simply true by an inductive argument. Initially all active vertices are stable and every time a vertex gets active it is stable since it has at least $r$ stable neighbors. This immediately implies that a dynamo whose all active vertices are stable is a monotone dynamo. 



\begin{definition}[torus]\label{def:torus}
The $d-$dimensional \emph{torus} is the graph $\mathbb{T}_n^d=(V,E)$ with vertex set $V:=\{(x_1,\dots,x_d):1\leq x_1,\dots,x_d\leq n\}$ and edge set $$E:=\{(x,x')\in V\times V: |x_j-x'_j\mod{n}|= 1~\text{for some $j$ and}~x_k=x'_k~\forall\ 1\leq k\neq j\leq n\}.$$
\end{definition}

The ``mod $n$'' means that the term is reduced modulo $n$. We always assume $d$ is a fixed number as we let $n$ tend to infinity. Note that $\mathbb{T}_n^d$ is a $2d-$regular graph with $n^d$ vertices. Moreover, given a vertex, we can obtain its neighbors simply by adding $\pm 1$ (mod $n$) to any of its coordinates. 

\section{Results}
We divide our results into two main parts, namely for large $r$ (i.e., $d+1\leq r\leq 2d$) and small $r$ (i.e., $1\leq r\leq d$), which are presented in Section~\ref{sec:r_large} and Section~\ref{sec:small_r}, respectively. We then extend our results to the majority process in Section \ref{section:majority_dynamo}. 
As a first step, we present an important relation between the size of a smallest dynamo in $r-$BP and reversible $r-$BP on bipartite graphs. 
\begin{theorem}\label{thm:bipartite_reversible_vs_irreversible}
Let $G=(V\cup U,E)$ be a bipartite graph of minimum degree at least $r$. Then $\overleftarrow{m}(G,r)\geq 2m(G,r)$.
\end {theorem}
\begin{proof} Note that the minimum degree condition is required because otherwise even a fully active configuration is not a dynamo. Let $\omega$ be a dynamo in $G$ under reversible $r-$BP. We construct a dynamo $\omega'$ of size $|\omega'|\leq |\omega|/2$ for $r-$BP. Without loss of generality, assume the number of active vertices in $V$ and configuration $\omega$ is at most $|\omega|/2$. Let $\omega'$ be the configuration where the state of each vertex in $V$ is identical to the one in $\omega$ and all vertices in $U$ are inactive. It suffices to show $\omega'$ is a dynamo in $r-$BP.
To this end, we use the following observation. Let $t'\in\N$ and consider the states of vertices in $U$ at time $t'$. Since all vertices in $U$ are only connected to vertices in $V$, we observe that under reversible $r-$BP rule, the states of vertices in $U$ at time $t'$ depend only on those in $V$ at time $t'-1$. Similarly, the states of the vertices in $V$ at time $t'+1$ depend only on the states of the vertices in $U$ at time $t'$. Hence, if $t$ is even then the states of the vertices in $V$ at time $t$ depend only on the initial states (time zero) of vertices in $V$. Clearly, $\omega$ is a dynamo also in $r-$BP; let $T$ be the first even time in which the whole graph is active in $r-$BP. By putting this point in parallel with the aforementioned observation, if $r-$BP process starts from $\omega'$, all vertices in $V$ will be active at time $T$ because the states of the vertices in $V$ only depend on their initial states, which are the same in $\omega$ and $\omega'$. In the next round of $r-$BP process induced by $\omega'$, all vertices in $V$ stay active and all vertices in $U$ are activated since they all have at least $r$ neighbors in $V$. This implies that $\omega'$ is a dynamo of size at most $|\omega|/2$ in $r-$BP. 
\qed
\end{proof}

In the following sections, we will repeatedly make use of the following notations for vertex sets on $\mathbb{T}_n^d$. For $i=0,1$ we let $A_{i}:=\{x\in\mathbb{T}_n^d:|x|\equiv i\mod{2}\}$ be the set of all vertices in the torus whose coordinates sum to an even or odd number, respectively. If $n$ is even then this gives a bipartition of $\mathbb{T}_n^d$ into a ``checker board-pattern''. We sometimes abuse the notation by writing the name of a graph instead of its vertex set. For $j=d,d+1,\dots,dn$ we denote the set of vertices whose coordinates sum to $j$ as $B_j:=\{x\in\mathbb{T}_n^d:|x|=j\}$. Note that $B_j\subset A_{i}$ if and only if $j\equiv i$ mod 2.
Next, to a given $r\leq d$ we let
\begin{equation}\label{eq:k_r}K(r):=\{(j_1,\dots,j_{d-r+1}): 1\leq j_1<\dots<j_{d-r+1}\leq d\}.\end{equation} 
Note that $|K(r)|=\binom{d}{d-r+1}=\binom{d}{r-1}$. For $d\geq 2$, $2\leq r\leq d$ and $k\in K(r)$ we define 
\begin{equation}\label{eq:smalldynamoconstruction}
    T(k):=\{x\in\mathbb{T}_n^d:x_j=1~\forall j=k_1,\dots,k_{d-r+1}\}~~\text{  and  }~~D(r):=\bigcup_{k\in K(r)} T(k).
\end{equation}
The vertices of $T(k)$ induce an $(r-1)-$dimensional sub-torus in $\mathbb{T}_n^d$. As an example, if $d=r=2$ then $D(r)$ is simply the union of the first ``row'' and the first ``column'' of $\mathbb{T}_n^2$.
\subsection{(Monotone) dynamos for large r}\label{sec:r_large}
As we discussed, Balister et al.~\cite{balbol09} proved $m(\mathbb{T}_n^d,r)=(1-\frac{d}{r})n^d\pm\Theta(n^{d-1})$ for large $r$.
Building on Theorem~\ref{thm:bipartite_reversible_vs_irreversible}, we first prove $2(1-\frac{d}{r})n^d-\Theta(n^{d-1})\le \overleftarrow{m}(\mathbb{T}_n^d,r)$ in Theorem~\ref{thm:irrev_vs_rev_threshold_dynamo_on_torus}. We then show $\overleftarrow{m}_{mon}(\mathbb{T}_n^d,r)\leq2(1-\frac{d}{r})n^d+\Theta(n^{d-1})$ in Theorem~\ref{thm:threshold_dynamos_large_r} by constructing a monotone dynamo of this size in reversible $r-$BP. Therefore, by applying $\overleftarrow{m}(\mathbb{T}_n^d,r)\le \overleftarrow{m}_{mon}(\mathbb{T}_n^d,r)$, we have both $\overleftarrow{m}(\mathbb{T}_n^d,r)$ and $\overleftarrow{m}_{mon}(\mathbb{T}_n^d,r)$ are equal to $2(1-\frac{d}{r})n^d\pm\Theta(n^{d-1})$.
\begin{theorem}\label{thm:irrev_vs_rev_threshold_dynamo_on_torus}
For all $d\geq 1$ and $d+1\leq r\leq 2d$ it holds $2(1-\frac{d}{r})n^d-\Theta(n^{d-1})\le \overleftarrow{m}(\mathbb{T}_n^d,r)$.
\end{theorem}
\begin{proof}
The statement is true for $r=2d$ because the only dynamo is the configuration in which all vertices are active. 

When $n$ is even then $\mathbb{T}_n^d$ is bipartite, and we also know that $m(\mathbb{T}_n^d,r)\geq(1-\frac{d}{r})n^d-\Theta(n^{d-1})$ (see Table~\ref{Table 1}). Thus, applying Theorem~\ref{thm:bipartite_reversible_vs_irreversible} implies our claim. 

For odd $n$, consider the torus $\mathbb{T}_{n+3}^d$. Let $\omega'\in\{0,1\}^{n^d}$ be a smallest dynamo in $\mathbb{T}_n^d$ under reversible $r-$BP. Let $\omega$ be the configuration in $\mathbb{T}_{n+3}^d$ whose entries are taken from $\omega'$ for the vertices in $[n]^d$, and the remaining vertices in $C:=[n+3]^d\backslash [n]^d$ are active. Note that the induced subgraph on $C$ has minimum degree $2d-1$, so all vertices in $C$ stay active in reversible $r-$BP for all $r\leq 2d-1$. Since $\omega'$ is a dynamo for reversible $r-$BP in $\mathbb{T}_{n}^d$ we have that $\omega$ is a dynamo in $\mathbb{T}_{n+3}^d$; this is because vertices in $C$ remain stable and so those vertices in $[n+3]^d\backslash C$ have at least as many active neighbors in every round of reversible $r-$BP process as they had in the corresponding process in $\mathbb{T}_n^d$. Since $|C|\leq 3d(n+3)^{d-1}=\Theta(n^{d-1})$, we have $|\omega|\le |\omega'|+\Theta(n^{d-1})$. Furthermore, since $n+3$ is even, the size of any dynamo in $\mathbb{T}_{n+3}^d$, including $\omega$, is lower-bounded by $2(1-\frac{d}{r})(n+3)^d-\Theta(n^{d-1})$. Therefore, $2(1-\frac{d}{r})(n+3)^d-\Theta(n^{d-1})\leq |\omega'|$, where $\omega'$ is a smallest dynamo in $\mathbb{T}_n^d$.
\qed
\end{proof}

We now provide a (monotone) dynamo in reversible $r-$BP of size $2(1-\frac{d}{r})n^d+\Theta(n^{d-1})$, where we rely on some construction from~\cite{balbol09}.
\begin{theorem}\label{thm:threshold_dynamos_large_r}
For $d\geq 1$ and $d+1\le r\leq 2d$ it holds $\overleftarrow{m}_{mon}(\mathbb{T}_n^d,r)\le 2(1-\frac{d}{r})n^d+\Theta(n^{d-1}).$
\end{theorem}
\begin{proof} 
The configuration in which all vertices are active is of size $n^d$ and is a dynamo for $r=2d$. Thus, the statement holds for $r=2d$ and we may assume $d\geq 2$ with impunity.

Define $H:=\bigcup_{j=1}^d\{x\in\mathbb{T}_n^d: x_j\leq 2\}$, which is the union of the vertices in $d$ ``perpendicular'' pairs of neighboring $(d-1)-$dimensional tori in $\mathbb{T}_n^d$ (so that $|H|\leq 2dn^{d-1}$). Every vertex in $H$ has at least $2d-1$ neighbors in $H$, so if we make the vertices in $H$ active initially then they stay active in reversible $r-$BP for all $r\leq 2d-1$, i.e. they are stable. 

In Lemma \ref{lem:construction_of_S} in the appendix, we discuss that there exists a vertex set $S\subset[n]^d$ with the following properties. (The construction of $S$ and the proof of the respective properties come from the proof of Theorem 26 in~\cite{balbol09}).
\begin{enumerate}
    \item Every vertex in $S$ has precisely $r$ neighbors in $S$.
    \item If $x\in B_j\backslash (S\cup H)$ then $x$ has precisely $r-d$ neighbors in $S\backslash B_{j-1}$.
    \item $|S|\leq 2\big(1-\frac{d}{r}\big)n^d+\Theta(n^{d-1})$.
\end{enumerate}
A depiction of the set $H\cup S$ for the case $d=2$ and $r=3$ is given in Figure~\ref{fig:figure} left, where black and white respectively represent active and inactive. We claim the configuration $\omega$ in which all vertices in $H\cup S$ are active is a monotone dynamo of our desired size. Firstly, $|\omega|\le |S|+|H|\le2(1-\frac{d}{r})n^d+\Theta(n^{d-1})$. Hence, it remains to prove $\omega$ is a monotone dynamo. Since all initially active vertices, namely the sets $H$ and $S$, are stable then by Lemma~\ref{stable-lemma} every vertex which gets active during the process is stable. Thus, to show $\omega$ is a monotone dynamo, one must show that every vertex gets active eventually. For $t\geq 0$ let $U_t:=S\cup H\cup\Big(\bigcup_{j=d}^{(3d-1)+t} B_j\Big)$. Note that $B_j\subset H$ for all $j\le 3d-1$. Thus, $U_0=S\cup H$ which is fully active in $\omega$. Given $U_0$ is active at time $t=0$, we prove by induction that the vertices in $U_t$ are stable at time $t$ for all $t>0$. Let $x\in B_{j}\subset U_{t}$ with $j=(3d-1)+t$ for some $t> 0$. If $x$ is active then there is nothing to prove, so suppose otherwise. Then $x\in B_{j}\backslash (H\cup S)$ because $H$ and $S$ are stable. $x$ has precisely $d$ neighbors in $B_{j-1}$ because decreasing a coordinate of $x$ by one gives an element in $B_{j-1}$. Moreover, $x$ has precisely $r-d$ neighbors in $S\backslash B_{j-1}$. By the induction hypothesis, all vertices in $U_{t-1}$ are active, so in particular, all vertices in $B_{j-1}$ are active. So $x$ has $d+(r-d)=r$ active neighbors and becomes active at time $t$. This shows that $U_{t}$ is stable at time $t$. Since $U_{dn}$ is equal to the whole vertex set, $\omega$ is a dynamo.
\qed
\end{proof}

\begin{figure}[!ht]
\centering
\includegraphics[scale=0.4]{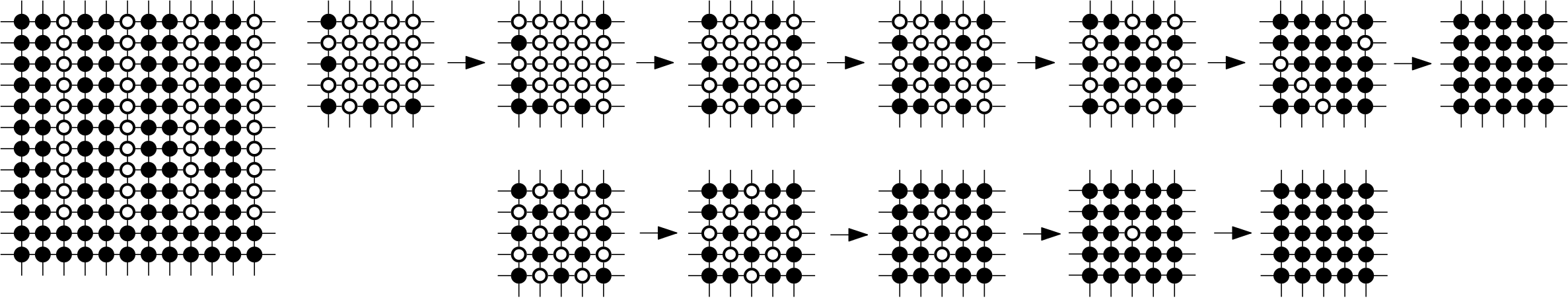}
\caption{Left: a monotone dynamo in $\mathbb{T}_n^2$ and reversible $3-$BP. Right: the top row illustrates the process induced by the construction from Theorem \ref{thm:odd_n_torus} for $d=r=2$, and the bottom row illustrates how the set $A_{0}$ from Lemma~\ref{lem:blinking_odd_n_dynamo} activates the whole torus for odd $n$ under reversible $2-$BP.}
\label{fig:figure}
\end{figure}
 In this section, we applied the fact that any lower bound on $\overleftarrow{m}(\mathbb{T}_n^d,r)$ is also a lower bound for $\overleftarrow{m}_{mon}(\mathbb{T}_n^d,r)$. However, actually we can directly prove $\overleftarrow{m}_{mon}(G,r)\geq 2(1-\frac{\Delta}{2r})|G|$ for any $\Delta-$regular graph $G$. The proof uses ideas similar to the one of Theorem 25 in~\cite{balbol09} and is given in Lemma~\ref{lem:monotone_dynamo_threshold} in the appendix, Section~\ref{regular graphs}. 
\subsection{(Monotone) dynamos for small r}\label{sec:small_r}
In this section we study (monotone) dynamos in $\mathbb{T}_n^d$ for $1\le r\leq d$. We first give the minimum size of $r-$BP dynamos and monotone reversible $r-$BP dynamos. We present our constructive upper bounds in Theorem~\ref{thm:good_small_r_dynamo}, and relying on prior results by Morrison and Noel~\cite{morrison2018extremal} we provide matching lower bounds. Then, building on these results we study the minimum size of reversible $r-$BP dynamos.
\begin{theorem}\label{thm:good_small_r_dynamo}
For all $d\geq 1$ and $1\leq r\leq d$ it holds

(i) $\overleftarrow{m}_{mon}(\mathbb{T}_n^d,r)\le \frac{2}{r}\binom{d}{r-1}n^{r-1}+\Theta(n^{r-2})$

(ii) $m(\mathbb{T}_n^d,r)\leq \frac{1}{r}\binom{d}{r-1}n^{r-1}+\Theta(n^{r-2}).$
\end{theorem}
To prove Theorem~\ref{thm:good_small_r_dynamo}, we first show a property of set $D(r)$ from (\ref{eq:smalldynamoconstruction}). 
\begin{lemma}\label{lem:small_r_dynamo_upper_bound}
Let $2\leq r\leq d$ and $\omega$ be a configuration on $\mathbb{T}_n^d$ where the active vertices are precisely $D(r)$. Then, $\omega$ is a monotone dynamo in reversible $r-$BP, and consequently $r-$BP.
\end{lemma}
\begin{proof}
By the construction of $D(r)$, every vertex in $D(r)$ has at least $2(r-1)\geq r$ neighbors inside $D(r)$, so if we make all vertices in $D(r)$ active then it is stable under reversible $r-$BP.
We now show $\omega$ is a monotone dynamo. Since all active vertices in $\omega$ are stable, by Lemma~\ref{stable-lemma} every vertex which gets active during the process is stable; thus, we must show each vertex eventually gets active. First of all, it holds $B_j\subset D(r)$ for all $d\leq j\leq d+(r-1)$ because for $j$ in this range every vertex in $B_j$ has at least $d-r+1$ coordinates with value one by the pigeonhole principle.
For $t\geq 0$ let us set $D_t:=D(r)\cup\Big(\bigcup_{j=d}^{t+(d+r-1)}B_j\Big).$
Then, we have $D_0=D(r)$, and our goal is to show that after $t\ge 0$ rounds of reversible $r-$BP the set $D_t$ is stable. For $D_0$ this holds by assumption, so we now perform induction. Suppose that $D_{t-1}$ is active for some $t-1\geq 0$. Set $j:=t+(d+r-1)$ and let $x\in B_{j}\subset D_{t}$. If $x\in D(r)$ then $x$ is already active, so suppose otherwise. When $x\not\in D(r)$, the vertex $x$ can have at most $d-r$ coordinates that are one. In particular, at least $r$ coordinates are greater than one. Decreasing any of them by one now gives a vertex in $B_{j-1}$, and every vertex in $B_{j-1}$ is stable by the induction hypothesis because $B_{j-1}\subset D_{t-1}$ by definition. So $x$ has $r$ active neighbors and therefore becomes active under reversible $r-$BP rule. 
\qed
\end{proof}

\begin{proof}\emph{(Theorem \ref{thm:good_small_r_dynamo})}
The statement is true for $r=1$ because for any connected graph a configuration with a single active vertex (analogously with two adjacent active vertices) is a monotone dynamo under $r-$BP (resp. reversible $r-$BP).

(i) Recall the notation from (\ref{eq:k_r}) and (\ref{eq:smalldynamoconstruction}). By Lemma \ref{lem:small_r_dynamo_upper_bound}, it suffices to construct a configuration $\omega$ that makes all vertices in $D(r)$ active monotonically. Let $\omega_r\in\{0,1\}^{n^{r-1}}$ be a smallest monotone dynamo in $\mathbb{T}_n^{r-1}$ in reversible $r-$BP for $r\geq 2$. From Theorem \ref{thm:threshold_dynamos_large_r} we know that $|\omega_r|\leq \frac{2}{r}n^{r-1}+\Theta(n^{r-2})$.
We now turn to the construction of a monotone dynamo for reversible $r-$BP in $\mathbb{T}_n^d$. For $k\in K(r)$, let $\omega_r^{(k)}$ be a smallest monotone dynamo under reversible $r-$BP in the $(r-1)$-dimensional sub-torus induced by $T(k)$.
We now define $\omega\in\{0,1\}^{n^d}$ as follows. 
Each vertex is active in $\omega$ if it is active in $\omega_r^{(k)}$ for some $k\in K(r)$ and inactive otherwise.
Since $|K(r)|=\binom{d}{r-1}$ and $|\omega_r^{(k)}|\leq \frac{2}{r}n^{r-1}+\Theta(n^{r-2})$ for any $k\in K(r)$, we obtain $|\omega|\leq |K(r)|\cdot|\omega_r^{(k)}|\leq \frac{2}{r}\binom{d}{r-1}n^{r-1}+\Theta(n^{r-2})$. Moreover, since $\omega_r^{(k)}$ is a monotone dynamo under reversible $r-$BP in the sub-torus induced by $T(k)$ for any $k\in K(r)$, all the vertices in $D(r)=\bigcup_{k\in K(r)} T(k)$ become active in reversible $r-$BP process induced by $\omega$. 

(ii) Redefine $\omega_r^{(k)}$ to be a smallest $r-$BP dynamo on the sub-torus induced by $T(k)$, which is of size at most $\frac{1}{r}n^{r-1}+\Theta(n^{r-2})$ (see Table~\ref{Table 1}). The proof then goes analogously.
\qed
\end{proof}


The tightness of the above upper bound for $r-$BP follows from Theorem~\ref{thm:hatami}, which is a result from~\cite{morrison2018extremal}. Furthermore, we provide Lemma~\ref{lem:dynamo_mon_vs_irreversible} which combined with Theorem~\ref{thm:hatami} prove the tightness of our construction for monotone dynamos in reversible $r-$BP.
\begin{theorem}[\cite{morrison2018extremal}]\label{thm:hatami}
For all $d\ge 1$ and $1\leq r\leq d$ it holds that $m(\mathbb{T}_n^d,r)\geq \frac{1}{r}\binom{d}{r-1}n^{r-1}$.
\end{theorem}
\begin{lemma}\label{lem:dynamo_mon_vs_irreversible}
For $d\geq 1$ and $1\leq r\leq 2d$ let $\omega$ be a monotone dynamo in $\mathbb{T}_n^d$ and reversible $r-$BP. Then there exists an $r-$BP dynamo $\omega'$ of size $|\omega'|\leq |\omega|/2+\Theta(|\omega|/n)$.
\end{lemma}
\begin{proof}
Let $W$ be the set of active vertices in $\omega$ and define $C:=\{x\in\mathbb{T}_n^d:\exists j\text{ s.t. }x_j=n\}$ to be the ``border vertices'' of $\mathbb{T}_n^d$, which are less than $dn^{d-1}$. So by shifting coordinates we may assume that $|C\cap W|\le \Theta(|\omega|/n)$. Moreover, there is an $i\in\{0,1\}$ such that $|W\cap A_{i}|\leq|\omega|/2$, without loss of generality $i=0$. Now, let $\omega'$ be the configuration in which all vertices in $W\cap (A_{0}\cup C)$ are active. Clearly, $|\omega'|\leq |\omega|/2+\Theta(|\omega|/n)$. It remains to show that $\omega'$ is a dynamo under $r-$BP. It suffices to show that after one round of $r-$BP by starting from configuration $\omega'$, all active vertices in $\omega$ (i.e., all vertices in $W$) will be active because $\omega$ is a dynamo also in $r-$BP. All vertices in $W\cap (A_{0}\cup C)$ stay active, and all vertices in $W\cap (A_{1}\setminus C)$ become active since they have at least $r$ active neighbors in $\omega'$. This is true because firstly by the monotonicity of $\omega$, each active vertex in $\omega$ is stable, meaning it has at least $r$ active neighbors. Furthermore, each vertex in $A_{1}\setminus C$ has all its neighbors in $A_{0}\cup C$. Thus, each vertex in $W\cap (A_{1}\setminus C)$ has $r$ active neighbors in $W\cap (A_{1}\setminus C)$, i.e. $\omega'$. 
\qed
\end{proof}

So far we provided tight bounds on the minimum size of monotone dynamos in $r-$BP and reversible $r-$BP. In the rest of this section, we focus on dynamos in reversible $r-$BP. We show that $\overleftarrow{m}(\mathbb{T}_n^d,r)=\frac{2}{r}\binom{d}{r-1}n^{r-1}\pm \Theta(n^{r-2})$ for even $n$ and $\overleftarrow{m}(\mathbb{T}_n^d,r)=\frac{1}{r}\binom{d}{r-1}n^{r-1}\pm\Theta(n^{r-2})$ for odd $n$. For even $n$, by Theorem~\ref{thm:good_small_r_dynamo} (i) we have a monotone dynamo, which is also a dynamo, of desired size. Since $\mathbb{T}_n^d$ is a bipartite graph for even $n$, Theorem~\ref{thm:bipartite_reversible_vs_irreversible} and Theorem~\ref{thm:hatami} together provide a matching lower bound. For odd $n$, the lower bound is immediate from Theorem~\ref{thm:hatami} because $\overleftarrow{m}(G,r)\geq m(G,r)$. Thus, it only remains to prove there is a dynamo of size $\frac{1}{r}\binom{d}{r-1}n^{r-1}+\Theta(n^{r-2})$ for odd $n$, which we do in Theorem~\ref{thm:odd_n_torus}. In order to prove this theorem let us first state Lemma~\ref{lem:blinking_odd_n_dynamo} and Lemma \ref{lem:from_w_r_to_a0mod2}.
\begin{lemma}\label{lem:blinking_odd_n_dynamo}
For $\mathbb{T}_n^d$ with odd $n$ and $1\leq r \leq d$, let $\omega$ be a configuration in which $A_0$ is fully active. Then $\omega$ is a dynamo under reversible $r-$BP and for any $1\leq r\leq d$.
\end{lemma}
\begin{proof}
Let $C(l):=\{x\in\mathbb{T}_n^d:\min_{c\in\{1,n\}^d}|x-c|=l\}$ for $0\leq l\leq d\lfloor n/2\rfloor$. The set $C(0)$ corresponds to all ``corner'' vertices and $C(l)$ are all vertices whose distance to the closest corner vertex is precisely $l$. We first assume that $d$ is even. Then for $x\in C(0)$, $|x|\equiv 0\mod{2}$ because $n$ is odd and $d$ is even; thus, for any $l\geq 0$ and any $x\in C(l)$, $|x|\equiv l\mod{2}$ by definition. Then, we have $C(l)\subset A_{i}$ for $i\in\{0,1\}$ if $l\equiv i\mod{2}$. 
We show that starting from $\omega$ the sets $C(l)$ become stable one by one in reversible $r-$BP process for $r=d$, and therefore for any $r\leq d$ (see Figure~\ref{fig:figure}, right). Firstly, $C(0)\subset A_{0}$ is stable under $\omega$ because each vertex in $C(0)$ has $d$ neighbors in $C(0)$. For every vertex $x=(x_1,\dots,x_d)\in C(0)$ and any $j=1,\dots,d$, if $x_j$ is equal to $1$ or $n$, we may replace it with respectively $n$ and $1$ to get a neighboring vertex in $C(0)$. Now we show for any $t\geq 1$, given $C(t')$ is stable for all $t'<t$, the set $C(t)$ becomes stable. At time $t$ the set $A_{i}$ is active for $t\equiv i$ mod $2$. This is because if $A_{0}$ is active at time $t-1$ then every vertex in $A_{1}$ has at least $d$ active neighbors so becomes active at time $t$, and vice versa. This implies $C(t)$ is active at time $t$. To prove the stability, we show that every vertex in $C(t)$ has $d$ neighbors in $C(t-1)\cup C(t)$. Let $x\in C(t)$, then for every coordinate $x_j$, if $1<x_j<n$ then either increasing it or decreasing it by 1 (mod $n$) must give a neighbor in $C(t-1)$ by definition. On the other hand, if $x_j\in\{1,n\}$ then either increasing it or decreasing it (mod $n$) gives a neighbor in $C(t)$. So, $C(t)$ becomes stable at time $t$ 
 and since $\bigcup_{t=0}^{d\lfloor n/2\rfloor} C(t)=[n]^d$, we infer that $\omega$ is a dynamo.
 
If $A_{0}$ is active in the initial configuration then $A_{1}$ is active in the next configuration in reversible $r-$BP and $r\leq d$. Thus, for odd $d$ the same proof works 
by setting the next configuration as the base case since we have $C(l)\subset A_{i}$ for $i\in\{0,1\}$ if $l\equiv i+1\mod{2}$.
\qed
\end{proof}

\begin{lemma}\label{lem:from_w_r_to_a0mod2}
For $\mathbb{T}_n^d$ with $d\geq 1$ and $d+1\leq r\leq 2d$, let $\omega$ be a dynamo in reversible $r-$BP. Then, there exists a configuration $\omega'$  of size $|\omega'|\leq|\omega|/2+\Theta(n^{d-1})$ such that in reversible $r-$BP process induced by $\omega'$, set $A_{0}$ gets active at some point.
\end{lemma}

\begin{proof}
For $r=2d$ the claim is clear as a dynamo is of size $n^d$. So assume $r\leq 2d-1$, and let $W$ be the set of active vertices in $\omega$. Furthermore, let $C:=\{x\in\mathbb{T}_n^d:\exists j~\text{s.t.} ~x_j> n-2\}$, which is of size less than $2dn^{d-1}=\Theta(n^{d-1})$. $|W\cap A_i|\le |W|/2$ for $i=0$ or $i=1$, say $i=0$. We claim configuration $\omega'$, where exactly all vertices in $(W\cap A_{0})\cup C$ are active, eventually makes set $A_0$ active. We have $|\omega'|=|(W\cap A_{0})\cup C|\leq|\omega|/2+\Theta(n^{d-1})$. Note that since $\omega$ is a dynamo, any configuration in which $W\cup C$ is fully active is also a dynamo. Let $T$ be the first even time for which all vertices are active in reversible $r-$BP with the initial configuration where exactly vertices in $W\cup C$ are active. We show that in reversible $r-$BP process induced by $\omega'$, all vertices in $A_{0}$ are active at time $T$. First of all, the vertices in $C$ are stable in $\omega'$ under reversible $r-$BP for any $r\leq 2d-1$ because they induce a subgraph of minimum degree $2d-1$. 
So we only have to argue that all vertices in $A_{0}\backslash C$ are active at time $T$. 
Notice the states of the vertices in $A_1$ at time $t$ only depend on the states of the vertices of $A_0\cup C$ at time $t-1$, and the states of the vertices in $A_0$ at time $t+1$ only depend on the states of the vertices of $A_1\cup C$ at time $t$. Since the vertices in $C$ are always active and $T$ is even, by an inductive argument the states of the vertices in $A_{0}\backslash C$ at time $T$ depend only on those in $A_{0}\cup C$ at time zero. So reversible $r-$BP process induced by $\omega'$ makes active all vertices in $A_{0}$ at time $T$. 
\qed
\end{proof}

\begin{theorem}\label{thm:odd_n_torus}
In $\mathbb{T}_n^d$ with odd $n$ and for $1\leq r \leq d$, $\overleftarrow{m}(\mathbb{T}_n^d,r)\le \frac{1}{r}\binom{d}{r-1}n^{r-1}+\Theta(n^{r-2})$.
\end{theorem}

\begin{proof} 
For $r=1$, consider the configuration in which only the vertex $x=(1,\cdots,1)$ is active. After $\lfloor n/2\rfloor$ rounds two adjacent vertices $(\lfloor n/2\rfloor, 1, \cdots, 1)$ and $(\lfloor n/2\rfloor+1, 1, \cdots, 1)$ are active. Two adjacent active vertices are stable and make the whole torus active eventually.

Let now $r\geq 2$ and recall the definition of $D(r)$ and $K(r)$ from (\ref{eq:k_r}) and (\ref{eq:smalldynamoconstruction}). For some $k=(k_1,\dots,k_{d-r+1})\in K(r)$, let $\omega_r\in\{0,1\}^{n^{r-1}}$ be a smallest dynamo in the subgraph induced by $T(k)$ under reversible $r-$BP. We know from Theorem~\ref{thm:threshold_dynamos_large_r} that $|\omega_r|\le \frac{2}{r}n^{r-1}+\Theta(n^{r-2})$. By Lemma~\ref{lem:from_w_r_to_a0mod2} we obtain a configuration $\omega_r'$ from $\omega_r$ with $|\omega_r'|\leq|\omega_r|/2+\Theta(n^{r-2})$ such that configuration $\omega_r'$ results in the activation of the vertices in $T(k)\cap A_{0}$ after $T$ rounds for some $T\ge 0$. We repeat this construction for $T(k)$ for all $k\in K(r)$. Let $\omega$ be the configuration obtained by taking $\binom{d}{r-1}$ copies of such $\omega_r'$ for the vertices in $\bigcup_{k\in K(r)}T(k)$ and setting the remaining entries to zero. 
We claim $\omega$ is a dynamo of our desired size. Firstly, $|\omega|\leq |K(r)|\cdot|\omega_r'|\leq \frac{1}{r}\binom{d}{r-1}n^{r-1}+\Theta(n^{r-2})$. Furthermore, after $T$ rounds, all vertices in $D(r)\cap A_0$ are active because $D(r)=\bigcup_{k\in K(r)} T(k)$. We show that from a configuration where all vertices in $D(r)\cap A_0$ are active, the process makes $A_0$ fully active. Then, applying Lemma~\ref{lem:blinking_odd_n_dynamo} finishes the proof.

We now argue that if $D(r)\cap A_{0}$ is active under reversible $r-$BP then $A_{0}$ becomes active (see Figure~\ref{fig:figure}, right). First assume that $d$ is even. We argue that at time $T+t$ ($t\geq 0$) the set $W_t:=\bigcup_{j\geq d,j\equiv t\mod{2}}^{t+(d+r-1)} B_j$ is active; thus after at most $T+d(n-1)$ rounds, the set $A_{0}$ is eventually active. By definition, $B_j\subset D(r)\cap A_{0}$ for all even $j$ with $d\leq j\leq d+r-1$. This means that at all times $t\in\N$, if $d\leq j\leq d+r-1$ and $j\equiv t$ mod $2$ then set $B_j$ is active at time $T+t$ because $D(r)\cap  A_{0}$ and $D(r)\cap A_{1}$ are active in alternating rounds. 
It remains to show that $B_j$ is active at time $t$ for all $j$ with $(d+r-1)< j\leq t+(d+r-1)$ and $j\equiv t$ mod 2. 
For $t=0$, this holds as we argued. 
We now show by induction that $W_t$ is active at time $t$ for all $t\geq 1$. To see this, let $x\in B_{j+1}$ for $j$ with $d\leq j\leq t+(d+r-1)$ and $j+1\equiv t+1$ mod 2. We show that $x$ is active at time $T+t+1$. If $x\in D(r)$ then this is true so let $x\not\in D(r)$. Then $x$ has at least $r$ coordinates greater than one, so it has at least $r$ neighbors in $B_j$ and $B_j\subset W_t$ is active at time $t$. So $B_{j+1}\subset W_{t+1}$ is active at time $t+1$. Therefore, for all even $t\geq d(n-1)$ the set $A_{0}\subset W_t$ is active at time $T+t$. If $d$ is odd then the proof goes analogously, simply by replacing $t$ with $t':=t+1$ and switching even and odd for $d$.
\qed
\end{proof}
\subsection{Monotone majority dynamos}\label{section:majority_dynamo}
Regarding the size of a smallest dynamo in $\mathbb{T}_n^d$ and the majority model, it is known~\cite{balbol09} that $\overleftarrow{m}(\mathbb{T}_n^d,maj)=(1-\frac{d}{d+1})n^{d}\pm \Theta(n^{d-1})$. However, $\overleftarrow{m}_{mon}(\mathbb{T}_n^d,maj)$ is only known for $d=2$~\cite{flocchini04}. We prove $\overleftarrow{m}_{mon}(\mathbb{T}_n^d,maj)=(1-\frac{d}{d+2})n^d\pm\Theta(n^{d-1})$. We first provide a lower bound on the minimum size of monotone dynamos in $2d-$regular graphs, including $\mathbb{T}_n^d$.
\begin {lemma}\label{lem:majority_dynamo_lower_bound}
For a $2d-$regular graph $G$, $\overleftarrow{m}_{mon}(G,maj)\geq (1-\frac{d}{d+2})|G|$.
\end {lemma}
\begin{proof} 
Let $\omega$ be an arbitrary monotone dynamo in the majority model, and define $G_0$ and $G_1$ to be respectively the subgraph induced by the inactive vertices and active vertices in $\omega$. $G_0$ must not contain a subgraph with minimum degree $d$ or larger because such a subgraph would stay inactive under the majority rule. So every subgraph of $G_0$ has a vertex of degree no more than $d-1$. This implies that $e(G_0)\leq (d-1)|G_0|$ where $e(G_0)$ denotes the number of edges in $G_0$ (just repeatedly remove the vertex with minimum degree). Furthermore, for the number of edges between $G_0$ and $G_1$ we have $e(G_0,G_1)=2d|G_0|-2e(G_0)$. By putting the last two equations together, we have $e(G_0,G_1)\geq 2d|G_0|-2(d-1)|G_0|=2|G_0|$. On the other hand, by the monotonicity of $\omega$, no vertex in $G_1$ can have more than $d$ edges to $G_0$, which implies $d|G_1|\geq e(G_0,G_1)$. Therefore, $d|G_1| \geq e(G_0,G_1)\ge 2|G_0|$, which gives $|G_1|\geq \frac{2}{d}|G_0|$. By applying $|G|=|G_0|+|G_1|$, we have $|G_1|\geq (1-\frac{d}{d+2})|G|$.
\qed 
\end{proof}

We now show that this bound is tight on the torus (notice that it is not necessarily tight for all $2d-$regular graphs, for instance the complete graph). The idea of the proof is similar to Theorem~\ref{thm:threshold_dynamos_large_r} and relies on some construction inspired from~\cite{balbol09}.
\begin {theorem}\label{thm:majority_dynamo_construction}
For all $d\geq 1$ it holds that $\overleftarrow{m}_{mon}(\mathbb{T}_n^d,maj)\leq (1-\frac{d}{d+2})n^d+\Theta(n^{d-1}).$
\end {theorem}

\begin{proof} 
Let us first introduce a vertex set $S$. We distinguish between $d$ even and $d$ odd.
For even $d$, 
set $S:=\bigl\{x=(x_1,\dots,x_d)\in\mathbb{T}_n^d: x_1+2x_2+\dots + (d/2)x_{d/2} \mod{\frac{d+2}{2}}= 1\bigr\}.$
As an example, in $\mathbb{T}_n^2$, $S$ is simply the union of the vertices in every second  ``column''. If $d$ is odd, 
we set $S:=\bigl\{x\in\mathbb{T}_n^d:2x_1+4x_2+6x_3+\dots+(d+1)x_{\frac{d+1}{2}}\mod{(d+2)}\in\{1,2\}\bigr\}.$ Let again $H:=\bigcup_{j=1}^d\{x\in\mathbb{T}_n^d: x_j\leq 2\}$.
In both cases ($d$ even and odd), $S$ satisfies the following:
\begin{enumerate}
    \item Every vertex in $S$ has precisely $d$ neighbors in $S$.
    \item If $x\in B_j\backslash (S\cup H)$ then $x$ has precisely one neighbors in $S\backslash B_{j-1}$.
    \item $|S|\leq (1-\frac{d}{d+2})n^d+\Theta(n^{d-1})$.
\end{enumerate}
These properties of $S$ are derived similar to the proof of Theorem 26 in~\cite{balbol09} (for details see Lemma~\ref{lem:construction_of_S2} in the appendix). Now, we claim a configuration $\omega$ in which precisely all vertices in $S\cup H$ are active is a monotone dynamo. Since $|S|\leq (1-\frac{d}{d+2})n^d+\Theta(n^{d-1})$ and $|H|=\Theta(n^{d-1})$, $\omega$ has our desired size. It remains to argue that $\omega$ is a monotone dynamo. Since all vertices which are active in $\omega$ are stable then by Lemma~\ref{stable-lemma} each vertex which gets active during the process is stable. Furthermore, with the same argument as the proof of Theorem \ref{thm:threshold_dynamos_large_r}, one can show $B_j$s get active one by one.
\qed
\end{proof}

\subsection*{Acknowledgement.}
The authors are thankful to Bernd G\"artner for several stimulating conversations and the careful reading of the manuscript.
\bibliographystyle{acm}
\bibliography{references}

\section{Appendix}

\subsection{Proof of Lemma~\ref{lem:monotone_dynamo_threshold}}
\label{regular graphs}

\begin {lemma}\label{lem:monotone_dynamo_threshold}
Let $G$ be a $\Delta-$regular graph and $r\leq \Delta$. Then $\overleftarrow{m}_{mon}(G,r)\geq 2(1-\frac{\Delta}{2r})|G|$.
\end {lemma}
\begin{proof}
Let $\omega$ be a monotone dynamo in $G$ in reversible $r-$BP. Let $G_0$ denote the subgraph induced by all inactive vertices and $G_1$ the subgraph induced by the active vertices in $\omega$. Every vertex in $G_1$ has at least $r$ neighbors in $G_1$, since it must never become inactive. Consequently, no vertex in $G_1$ can have more than $\Delta-r$ edges to $G_0$, which implies that for the number of edges $e(G_0,G_1)$ between $G_0$ and $G_1$ it holds $e(G_0,G_1)\leq (\Delta-r)|G_1|$. On the other hand, if $G_0$ contains a subgraph of minimum degree at least $\Delta-r+1$ then no vertex in this subgraph ever becomes active. This would contradict $\omega$ being a dynamo. We infer that $(\Delta-r)|G_1|\geq e(G_0,G_1)\geq (2r-\Delta)|G_0|,$ where the second inequality comes from Lemma \ref{lem:dynamo_subgraph_lower_bound} below. Using that $|G_0|+|G_1|=|G|$ we obtain $r|G_1|\geq (2r-\Delta)|G|$, which gives the claim.
\qed
\end{proof}

\begin {lemma}[from Theorem 25 in \cite{balbol09}]\label{lem:dynamo_subgraph_lower_bound}
Let $G$ be a $\Delta-$regular graph with configuration $\omega$. Assume $G_0$ and $G_1$ be the subgraphs induced by the inactive and active vertices, and let $1\leq r\leq \Delta$. If $G_0$ contains no subgraph of minimum degree $\Delta-r+1$ or more then $e(G_0,G_1)\geq (2r-\Delta)|G_0|.$
\end {lemma}

\begin{proof}
By the assumption, every induced subgraph of $G_0$ has a vertex of degree at most $\Delta-r$. Starting with $G_0$, we can iteratively take away the vertices of smallest degree. Since in every one of the resulting subgraphs there is a vertex of degree at most $\Delta-r$, we conclude that $e(G_0)\leq(\Delta-r)|G_0|$. Using the handshake lemma, summation over all vertex degrees then gives
$e(G_0,G_1)\geq\Delta|G_0|-2e(G_0)\geq (2r-\Delta)|G_0|$.
\qed
\end{proof}

\subsection{Proof of Lemma~\ref{lem:construction_of_S}}
\label{set S}

\begin{lemma}[adapted from the proof of Theorem 26 in \cite{balbol09}]\label{lem:construction_of_S}
For $\mathbb{T}_n^d$ with $d+1\le r\leq 2d$, let $H:=\bigcup_{j=1}^d\{x\in\mathbb{T}_n^d: x_j\leq 2\}$. Then there exists a vertex set $S\subset \mathbb{T}_n^d$ with the following properties:
\begin{enumerate}
    \item Every vertex in $S$ has precisely $r$ neighbors in $S$.
    \item If $x\in B_j\backslash (S\cup H)$ then $x$ has precisely $r-d$ neighbors in $S\backslash B_{j-1}$.
    \item $|S|\leq 2\big(1-\frac{d}{r}\big)n^d+\Theta(n^{d-1})$.
\end{enumerate}
\end{lemma}

\begin{proof}
We distinguish $r$ even and $r$ odd.
\begin{itemize}
    \item {If $r$ is even, we let $$f_{even}(x):=x_1+2x_2+\dots + (\frac{r}{2}-1)x_{\frac{r}{2}-1} \mod{\frac{r}{2}}$$
    and define
    $S:=\{x\in\mathbb{T}_n^d: f_{even}(x) \in\{1,2,\dots,r-d\}\}.$
    If $x\in S$ then $x$ has exactly $2(\frac{r}{2}-(r-d))=2d-r$ neighbors that are not in $S$, so exactly $r$ neighbors inside $S$. 
    Let now $x\not\in S$. Observe that by construction of $f_{even},$ we have that $r-d$ out of the first $(\frac{r}{2}-1)$ coordinates of $x$ can be increased by one such that the resulting vertex lies in $S$. This is because the coefficients $1,2,\dots,r/2-1$ correspond to all nonzero residue classes mod $\frac{r}{2}$. Symmetrically, there are $r-d$ coordinates of $x$ that we can decrease by one to obtain an element of $S$. This simply means that every vertex that does not lie in $S$ has precisely $2(r-d)$ neighbors in $S$, and it is immediate that if $x\in B_j\backslash H$ for some $j$ then precisely half of these neighbors lie in $S\cap B_{j-1}$ (namely those we obtain by decreasing one coordinate of $x$) and the other $r-d$ neighbors lie in $S\backslash B_{j-1}$.}
    \item {If $r$ is odd, we let
    $$f_{odd}(x):=2x_1+4x_2+6x_3+\dots+(r-1)x_{\frac{r-1}{2}}\mod{r}$$
    and define
    $S:=\{x\in\mathbb{T}_n^d:f_{odd}(x)\in\{1,2,\dots, 2(r-d)\}\}.$
    As before, we show that every vertex in $S$ has exactly $r$ neighbors in $S$. Again, note that the coefficients $\pm 2$, $\pm 4$, $\dots,\pm (r-1)$ are precisely the non-zero residue classes $\mod{r}$. Then as above, $x$ has $(r-2(r-d))=2d-r$ neighbors outside $S$ and $r$ neighbors that lie in $S$.
    For $x\not\in S$, among the first $\frac{r-1}{2}$ coordinates there are $r-d$ coordinates $i$ such that for $x':=(x_1,\dots,x_{i-1},x_i+1,x_{i+1},\dots,x_d)$ we have $f(x')=f(x)+2i\mod{r}\in\{1,\dots,2(r-d)\}$, meaning $x'\in S$. Analogously, there are $r-d$ coordinates that we can decrease by $1$ to obtain neighbors of $x$ in $S$. All in all, $x$ again has $2(r-d)$ neighbors in $S$. As above, we note that if $x\in B_j\backslash H$ for some $j$ then precisely half of these neighbors lie in $S\cap B_{j-1}$, and the other $r-d$ neighbors lie in $S\backslash B_{j-1}$.}
\end{itemize}
It remains to estimate the size of $S$. Fixing the coordinates $x_2,\dots,x_d$ of a vertex leaves at most $\lceil n\frac{2(r-d)}{r}\rceil$ choices for $x_1$ to obtain a vertex $x\in S$, in both cases ($r$ even and $r$ odd). It follows that $|S|\leq \lceil n\frac{2(r-d)}{r}\rceil n^{d-1}\leq 2(1-\frac{d}{r})n^d+n^{d-1}$.
\qed
\end{proof}

\subsection{Proof of Lemma~\ref{lem:construction_of_S2}}
\label{majority}

\begin{lemma}\label{lem:construction_of_S2}
Let $H:=\bigcup_{j=1}^d\{x\in\mathbb{T}_n^d: x_j\leq 2\}$. Then there exists a vertex set $S\subset \mathbb{T}_n^d$ with the following properties:
\begin{enumerate}
    \item Every vertex in $S$ has precisely $d$ neighbors in $S$.
    \item If $x\in B_j\backslash (S\cup H)$ then $x$ has precisely one neighbors in $S\backslash B_{j-1}$.
    \item $|S|\leq (1-\frac{d}{d+2})n^d+\Theta(n^{d-1})$.
\end{enumerate}
\end{lemma}

\begin{proof}
We distinguish between $d$ even and $d$ odd. Let $x=(x_1,\dots,x_d)\in\mathbb{T}_n^d$ be a vertex.
\begin{itemize}
    \item If $d$ is even, we let $$f_{even}(x):=x_1+2x_2+\dots + (d/2)x_{d/2} \mod{\frac{d+2}{2}}$$
    and set
    $S:=\{x\in\mathbb{T}_n^d: f_{even}(x)= 1\}.$
    Let us analyze the properties of the subgraph induced by $S$. If $x\in S$ then adding $+1$ or $-1$ to one of the coordinates $x_{d/2+1},\dots,x_d$ gives a vertex in $S$. Therefore, every vertex in $S$ has $2\cdot d/2=d$ neighbors in $S$.
    Let now $x\not\in S$. Observe that by construction of $f_{even},$ one of the first $d/2$ coordinates of $x$ can be increased by one such that the resulting vertex lies in $S$. This is because the coefficients $1,2,\dots,d/2$ correspond to all nonzero residue classes mod $(d/2+1)$. Symmetrically, there is one coordinate of $x$ that we can decrease to obtain an element of $S$. This simply means that every vertex that does not lie in $S$ has two neighbors in $S$, and it is immediate that if $x\in B_j\backslash H$ for some $j$ then precisely one of these neighbors lies in $S\cap B_{j-1}$ (namely the one we obtain by decreasing one coordinate of $x$).
    \item If $d$ is odd, we let
    $$f_{odd}(x):=2x_1+4x_2+6x_3+\dots+(d+1)x_{\frac{d+1}{2}}\mod{(d+2)}$$
    and set
    $S:=\{x\in\mathbb{T}_n^d:f_{odd}(x)\in\{1,2\}\}.$
    As before, we show that the induced subgraph by $S$ is $d-$regular. Pick an $x\in S$ with $f_{odd}(x)=2$. It immediately follows by the same argument as above that $x$ has $2(d-\frac{d+1}{2})=d-1$ neighbors in $S$ that are obtained by adding $\pm 1$ to any of the $\frac{d-1}{2}$ last coordinates. We obtain one further neighbor of $x$ in $S$ by adding $+1$ to $x_{\frac{d+1}{2}}$, which results in a vertex $x'$ with $f_{odd}(x')=1$. Analogously, if $f_{odd}(x)=1$ then we subtract one from $x_{\frac{d+1}{2}}$ to get an $x'$ with $f_{odd}(x')=2$ and so $x'\in S$. Hence, every vertex in $S$ has $d$ neighbors in $S$.
    Furthermore, for $x\not\in S$ one of the first $\frac{d+1}{2}$ coordinates of $x$ must be such that adding one to it gives a neighbor in $S$. Analogously, one of these coordinates can be decreased by one to give another neighbor in $S$.
    As above, we note that if $x\in B_j\backslash H$ for some $j$ then $x$ has precisely one neighbor in $S\cap B_{j-1}$.
\end{itemize}
It remains to estimate $|S|$. Fixing the coordinates $x_2,\dots,x_d$ of a vertex leaves at most $\lceil n\frac{2}{d+2}\rceil$ choices for $x_1$ to obtain a vertex $x\in S$, in both cases ($d$ even and $d$ odd). It follows that $|S|\leq (1-\frac{d}{d+2})n^d+\Theta(n^{d-1})$.
\qed
\end{proof}

\end{document}